\documentclass[conference,letterpaper]{IEEEtran}

\usepackage{graphicx} 
\usepackage{amsmath}
\usepackage{amsfonts}
\usepackage{amsthm}
\usepackage[margin=1in]{geometry}
\usepackage{xcolor}
\usepackage{cite}
\usepackage{enumitem}

\newcommand{\rmd}{\mathrm{d}}
\newcommand{\bbE}{\mathbb{E}}\newcommand{\rme}{\mathrm{e}}

\newcommand{\sfD}{\mathsf{D}}

\newcommand{\supp}{{\mathsf{supp}}}

\newcommand{\chiSq}{\chi^2}
\newcommand{\TV}{\mathsf{TV}}

\newtheorem{thm}{Theorem}

\newtheorem{prop}{Proposition}

\newtheorem{lem}{Lemma}
\newtheorem{rem}{Remark}

\title{ Support Size of $\varepsilon$-Capacity-Achieving Inputs for the Amplitude-Constrained AWGN Channel }
\author{%
  \IEEEauthorblockN{Luca Barletta}
  \IEEEauthorblockA{Dipartimento di Elettronica, Informazione e Bioingegneria \\
                    Politecnico di Milano. Milan, Italy\\
                    Email: luca.barletta@polimi.it}
  \and
  \IEEEauthorblockN{Alex Dytso}
  \IEEEauthorblockA{Qualcomm Flarion Technology, Inc.\\ 
                    Bridgewater, NJ, USA\\
                    Email: odytso2@gmail.com}
}

\begin{document}

\maketitle

\begin{abstract}
We study the amplitude-constrained additive white Gaussian noise (AWGN) channel from the perspective of near-optimal input distributions. While it is known that the capacity-achieving input is discrete with finitely many mass points, the precise scaling of its support size as a function of the amplitude constraint remains an open problem. In this work, we instead consider the minimal support size required to achieve capacity up to an $\varepsilon$-gap. 
We introduce the quantity $K_\varepsilon(A)$, defined as the smallest support size among discrete inputs supported on $[-A,A]$ that achieves mutual information within $\varepsilon$ of capacity. We show that this relaxed formulation is significantly more tractable and admits sharp characterizations across different regimes of $\varepsilon$. In particular, when $\varepsilon$ decays polynomially with $A$, i.e., $\varepsilon = A^{-\beta}$ for $\beta \geq 1$, we establish that $K_\varepsilon(A) = \Theta(A\sqrt{\log A})$. For exponentially small gaps, we obtain bounds of order between $A\sqrt{\log A}$ and $A^{3/2}$. 
Our approach combines approximation-theoretic bounds for Gaussian mixtures with information-theoretic control of entropy via $\chi^2$-divergence, together with a wrapping argument that relates the problem to approximating the uniform distribution on the circle. Beyond the technical results, our framework provides a conceptual explanation for the variety of scaling laws observed in prior numerical studies, showing that these correspond to different regimes of $\varepsilon$-optimality rather than intrinsic properties of the exact optimizer.
\end{abstract}


\section{Introduction}
We consider an additive white Gaussian noise (AWGN) channel subject to a peak-power constraint. The
channel output is given by
\begin{equation}\label{eq:channel}
Y = X + Z,
\end{equation}
where the input random variable $X$ satisfies the constraint $|X| \le A$ almost surely (a.s.), and $Z$ is a standard normal random
variable independent of $X$.
The capacity under the amplitude constraint is
\begin{equation}\label{eq:cap_def}
C(A) := \sup_{P_X:\ \supp(P_X)\subset[-A,A]} I(X;Y).
\end{equation}
We denote by $X^*$ a capacity-achieving
input random variable and by $Y^*$ the corresponding induced output random variable. In general, both the
exact value of the capacity $C(A)$ and the precise structure of the capacity-achieving distribution $X^*$ remain
unknown.

\subsection{Problem Formulation}
In contrast to prior work focusing on exact optimality, we study the minimal support size required to achieve capacity up to an $\varepsilon$-gap. Specifically, for $\varepsilon > 0$, define
\begin{align}\label{eq:Keps_def}
K_\varepsilon(A)
:= \min \Big\{&|\supp(P_X)|:\ \supp(P_X)\subset[-A,A], \nonumber \\
& P_X \text{ discrete},\ I(X;Y)\ge C(A)-\varepsilon\Big\}.
\end{align}
While the precise scaling of $|\supp(P_{X^*})|$ remains unknown, this relaxed formulation turns out to be significantly more tractable. In particular, we are able to characterize the exact scaling of $K_\varepsilon(A)$ for a range of regimes. For example, when the gap decays polynomially with $A$, i.e., $\varepsilon = A^{-\beta}$ for some $\beta \geq 1$, we obtain a sharp characterization of $K_\varepsilon(A)$.

\subsection{Literature Review}
The literature on the amplitude constraint channel is large and we do not try to survey it fully and only mention key relevant results. For comprehensive review, interest readers are referred to~\cite{GaussianBoundsCard, dytso2017amplitude_Globecom,CISS2018} and references therein.

Studying the capacity of the amplitude-constrained AWGN channel is a classical problem in information theory, originating in the work of Shannon~\cite{Shannon:1948}. A fundamental result due to Smith~\cite{smith1969Thesis,smith1971information} shows that the capacity-achieving input distribution is discrete with finitely many mass points, in contrast to the average-power constrained setting where the capacity-achieving distribution is Gaussian.  Subsequent work further characterized the structure of the optimal input, including transition thresholds where binary and ternary constellations are optimal~\cite{sharma2010transition}. However, the precise scaling of the support size with the amplitude constraint remains unresolved: the best known non-asymptotic bounds, \cite{wang2025improvedlowerboundcardinality} and~\cite{GaussianBoundsCard}, place it between  $A\sqrt{\log A}$ and $A^2$, respectively.  

Somewhat intriguingly, numerical investigations have suggested a range of alternative asymptotic behaviors. In particular, Mattingly et al.~\cite{mattingly2018maximizing} reported an empirical scaling of order $A^{4/3}$ as $A \to \infty$, based on experiments involving not only the AWGN channel but also several non-Gaussian models, including the binomial channel and certain two-dimensional channels. A follow-up work by Abbott and Machta~\cite{abbott2019scaling} provided a heuristic, physics-inspired justification for this scaling. These findings coexist with Zhang’s spacing-based heuristic, which suggests growth of order $A\sqrt{\log A}$~\cite[pp.~91,~95,~96]{zhang1994discrete}, while earlier work conjectured linear scaling~\cite{GaussianBoundsCard}, a claim that has since been disproved~\cite{wang2025improvedlowerboundcardinality}.

We argue that this apparent diversity of scaling laws can be naturally interpreted through the lens of $\varepsilon$-optimality. Indeed, all numerical procedures implicitly operate with a finite tolerance, effectively computing inputs that are only $\varepsilon$-capacity-achieving for some algorithm-dependent $\varepsilon$. From this perspective, different observed scalings correspond to different regimes of $\varepsilon = \varepsilon(A)$, rather than intrinsic properties of the exact optimizer. This viewpoint also explains the well-known numerical sensitivity of the problem: in the large-$A$ regime, the optimal output distribution becomes nearly uniform in the interior, and deviations that distinguish competing inputs occur at a scale comparable to numerical precision. Consequently, implementations based on, e.g., the Blahut–Arimoto algorithm~\cite{blahut2003computation,arimoto1972algorithm}, which involve repeated numerical integrations of log-densities, are particularly susceptible to bias and instability. As a result, different numerical tolerances and methodologies can lead to markedly different empirical scaling laws.

In parallel, a large body of work has developed capacity bounds via entropy methods, duality, and estimation-theoretic representations; see~\cite{dytso2019amplitude} for a comprehensive review. 

\subsection{Outline and Contributions}
Section~\ref{sec:tools} collects the main technical tools used throughout the paper, including stability bounds for entropy, approximation results for Gaussian mixtures, and properties of the wrapping operation. These ingredients form the backbone of both the achievability and converse arguments.

Section~\ref{sec:main_results} presents the main results together with their derivations. In particular, we obtain sharp bounds on $K_\varepsilon(A)$ and provide a complete characterization in the regime where $\varepsilon$ decays at most polynomially in $A$. We also discuss the behavior in the exponential regime and highlight the transition in scaling.

Section~\ref{sec:conclusion} concludes the paper with a summary of the main findings and a discussion of open problems.

We conclude this section by presenting relevant notation. 

\subsection{Notation}
Throughout the paper, the deterministic scalar quantities are denoted by lower-case letters and random variables are denoted by uppercase letters.  

We denote the distribution  of a random variable $X$ by $P_{X}$. The support set of $P_X$ is denoted and defined as
\begin{align}
\supp(P_{X})= &\left\{ x:  \text{ for every open set $ \mathcal{D} \ni x $} \right. \nonumber \\ 
& \left.\quad \text{we have that $P_{X}( \mathcal{D})>0$}  \right\}. 
\end{align} 
The notation $| \cdot |$, depending on the context, denotes either absolute value or cardinality of the set.  All logarithm are taken with base $\rme$.  The density of a standard normal will be denoted by $\varphi$. 

 Given two probability distributions $P$ and $Q$ with probability densities functions (pdfs) $p$ and $q$, respectively, we will require the following distances: 
\begin{align}
  &\text{ Total Variation:}  &\mathsf{TV}(P \| Q) &= \frac{1}{2} \int | p(x) -q(x)| \rmd x, \label{eq:tv_def} \\
  &\text{ Relative Entropy:} &\sfD(P \| Q) &=  \int p(x) \log \frac{p(x)}{q(x)} \rmd x,\\
  &\text{ $\chi^2$ Divergence:} &\chi^2(P \| Q) &=  \int \frac{( p(x) -q(x))^2}{q(x)} \rmd  x ,\label{eq:xi_def}
\end{align}
with the understanding that the relative entropy and $\chi^2$ are equal to infinity if $P$ is not absolutely continuous with respect to $Q$.  

\section{Tools and Preliminaries}
\label{sec:tools}

In this section, we collect several technical ingredients that underlie our analysis. 
At a high level, our approach proceeds by approximating the output distribution induced by the capacity-achieving input using finite Gaussian mixtures, and then quantifying the resulting loss in mutual information. 
This leads to three main components: 
(i) a stability bound for entropy in terms of $\chi^2$-divergence, 
(ii) approximation guarantees for Gaussian mixtures, and 
(iii) a wrapping argument that allows us to compare distributions to the uniform law on the circle.

\subsection{Entropy Loss via $\chi^2$}

The first ingredient is a quantitative stability bound for differential entropy under $\chi^2$-perturbations. 
This allows us to convert approximation guarantees at the level of densities into bounds on mutual information.

\begin{lem}[Entropy loss controlled by $\chi^2$]\label{lem:entropy_stability}
Let $f,g$ be densities on $\mathbb R$ such that $\chi^2(g\|f)<\infty$ and $\int (\log f)^2 f<\infty$. Then,
\begin{equation}\label{eq:entropy_stability}
h(f)-h(g)\ \le  \|\log f\|_{L^2(f)}\sqrt{\chi^2(g\|f)}\ +\ \chi^2(g\|f),
\end{equation}
where $\|\log f\|_{L^2(f)}:=\big(\int (\log f)^2 f\big)^{1/2}$.
\end{lem}

\begin{proof}
Write $u:=\frac{g}{f}-1$, so that $g=(1+u)f$ and $\int u f = \int g-\int f=0$. Then $\chi^2(g\|f)=\int u^2 f$.
Compute
\begin{align*}
h(g)
&= -\int g\log g
= -\int (1+u)f\log\big((1+u)f\big) \\
&= -\int (1+u)f\log f \ - \ \int (1+u)f\log(1+u).
\end{align*}
Hence
\begin{align*}
h(f)-h(g)
&= -\int f\log f + \int (1+u)f\log f \\
&\quad+ \int (1+u)f\log(1+u) \\
&= \int u f \log f \ +\ \int (1+u)f\log(1+u).
\end{align*}
By Cauchy--Schwarz,
\begin{align*}
\Big|\int u f\log f\Big|
&\le \Big(\int u^2 f\Big)^{1/2}\Big(\int (\log f)^2 f\Big)^{1/2} \\
&= \|\log f\|_{L^2(f)}\sqrt{\chi^2(g\|f)}.
\end{align*}
For the second term, use $\log(1+u)\le u$ for all $u>-1$, so $(1+u)\log(1+u)\le (1+u)u=u+u^2$.
Integrating against $f$ yields
\begin{align*}
\int (1+u)f\log(1+u)&\le \int (u+u^2)f  \\
&= \int u f + \int u^2 f \\
&= \chi^2(g\|f).
\end{align*}
Combining the two bounds proves \eqref{eq:entropy_stability}.
\end{proof}

The key feature of Lemma~\ref{lem:entropy_stability} is that it provides  control of entropy loss in terms of $\chi^2$-divergence. 
In particular, small $\chi^2$-error directly translates into a small loss in mutual information, up to a multiplicative factor depending on $\|\log f\|_{L^2(f)}$.

We will apply this lemma in a setting where $f$ corresponds to the output density induced by the capacity-achieving input, and $g$ corresponds to an approximating Gaussian mixture. 
The next result provides a uniform bound on the prefactor $\|\log f\|_{L^2(f)}$ in this setting.

\begin{lem}\label{lem:logL2_bound}
Fix $A>0$. Let $Y = X+Z $ where $Z \sim \mathcal{N}(0,1)$ and $X\sim P_X$ be  supported on $[-A,A]$. Then, 
\begin{equation}\label{eq:logL2_bound}
\|\log f_{Y}\|_{L^2(f_Y)}\ \le\ \sqrt{10}(1+A^2).
\end{equation}
\end{lem}

\begin{proof}
Fix $y\in\mathbb R$. Since $\varphi(x) = (2\pi)^{-1/2}e^{-x^2/2}$ is strictly decreasing in $|x|$, for $\theta\in[-A,A]$ we have
$|y-\theta|\le |y|+A$, hence $\varphi(y-\theta)\ge \varphi(|y|+A)$. Averaging gives

\begin{align}
    f_Y(y)&=\int \varphi(y-\theta) \rmd P_X(\theta) \nonumber\\
    &\ge\ \varphi(|y|+A)=(2\pi)^{-1/2}\exp\!\Big(-\frac{(|y|+A)^2}{2}\Big).
\end{align}
Thus,
\begin{equation}
    -\log f_Y(y)\ \le\ \frac{(|y|+A)^2}{2} + \frac12\log(2\pi)=:\frac{(|y|+A)^2}{2}+c_0. \label{eq:bound_on_fy}
\end{equation}
Consequently,
\begin{align}
    \|\log f_{Y}\|_{L^2(f_Y)}^2 &=  \bbE \left[ (-\log f_Y(Y))^2 \right]\\
    & \le  \bbE \left[ \left( \frac{(|X+Z|+A)^2}{2}+c_0 \right)^2 \right] \label{eq:bounding_log}\\
    & \le 8(1+A^2)^2+2c_0^2 \label{eq:_Lp_bounds_1} \\
    &\le 10(1+A^2)^2
\end{align}
where~\eqref{eq:bounding_log} follows from the bound in \eqref{eq:bound_on_fy}; and \eqref{eq:_Lp_bounds_1} follows from sequential application of the inequality $(a+b)^2 \le 2 (a^2 +b^2) $, the bound $|X| \le A$, and the evaluation of the moments of $Z$. 
\end{proof}

Lemma~\ref{lem:logL2_bound} shows that the entropy sensitivity grows at most quadratically in $A$. 
Combined with Lemma~\ref{lem:entropy_stability}, this implies that achieving a small $\chi^2$-approximation error is sufficient to ensure near-optimality in mutual information.

\subsection{Finite-Mixture Approximations}
The second ingredient is a sharp approximation result for Gaussian mixtures. 
Recall that for any probability measure $P$ supported on $[-A,A]$, the induced output density takes the form
\begin{equation}
f_P(y):=\int_{-A}^A \varphi(y-\theta)\,\rmd P(\theta), \label{eq:mix_density}
\end{equation}
i.e., a Gaussian location mixture.

Our goal is to approximate such densities using mixtures supported on finitely many points. 
The following result, due to Ma, Wu, and Yang, provides near-optimal bounds on this approximation error.

\begin{lem}[Ma--Wu--Yang approximation theorem \cite{BestApporximationGaussianMixture}]\label{lem:MWY}
Let $A>0$ and let $\mathcal P_A^{\mathrm{Bdd}}$ denote the set of probability measures supported on $[-A,A]$.
Let $\mathcal P_m$ denote the set of probability measures supported on at most $m$ points.
Define the worst-case best approximation error
\[
E^\star(m,\mathcal P_A^{\mathrm{Bdd}},\chi^2)
:=\sup_{P\in \mathcal P_A^{\mathrm{Bdd}}}\ \inf_{Q\in \mathcal P_A^{\mathrm{Bdd}}\cap \mathcal P_m}\ \chi^2(f_Q\|f_P),
\]
where $f_P$ is the mixture density \eqref{eq:mix_density}.
Then there exists a universal constant $\kappa\ge 16e^3$ such that for all $m\in\mathbb N$ and $A>0$,
\begin{align}\label{eq:MWY}
& E^\star(m,\mathcal P_A^{\mathrm{Bdd}},\chi^2) \nonumber\\ &\le\
\begin{cases}
\exp\!\big(- \frac{m\log m}{A^2}\big), & m\ge \kappa A^2,\\[0.3em]
\exp\!\big(- \frac{\log\kappa}{4\kappa}\frac{m^2}{A^2}\big), & 3\sqrt{\kappa A}\le m\le \kappa A^2.
\end{cases}
\end{align}
\end{lem}

The key takeaway from Lemma~\ref{lem:MWY} is that Gaussian mixtures supported on $m$ points can approximate arbitrary mixtures supported on $[-A,A]$ with exponentially small $\chi^2$-error. 
Moreover, the approximation exhibits two distinct regimes:
\begin{itemize}
    \item a \emph{quadratic regime}, where the error behaves like $\exp(-m^2/A^2)$,
    \item and a \emph{large-$m$ regime}, where the error behaves like $\exp(-m\log m / A^2)$.
\end{itemize}

This dichotomy will directly translate into the two regimes in our main results. 
In particular, it is precisely this approximation behavior that determines the scaling of $K_\varepsilon(A)$.

\subsection{Wrapped Random Variables} 
The final ingredient is a wrapping argument, which was introduced in \cite{wang2025improvedlowerboundcardinality} that allows us to compare output distributions to the uniform distribution on the circle. 
This plays a key role in the converse, where we lower bound the support size by quantifying how far the induced output can be from uniformity.

For $B>0$, define the wrapping map $\langle\cdot\rangle_B:\mathbb R\to[-\pi,\pi)$ by
\begin{equation}
\langle W\rangle_B:=\frac{\pi}{B}(W\bmod 2B),
\end{equation}
where $W \text{ mod } 2B := W -2 B \left \lfloor \frac{W+B}{2B} \right \rfloor$.

The wrapping operation has several useful properties summarized below. 
\begin{prop}
    \label{prop:summary_wrapping_operation} \text{ }
    \begin{enumerate}
    \item  (Wrapped density formula) 
If $W$ has density $f_W$, then $\langle W\rangle_B$ has density
\begin{equation}\label{eq:wrap_pdf}
f_{\langle W\rangle_B}(\theta)=\frac{B}{\pi}\sum_{k\in\mathbb Z} f_W\!\Big(\frac{B}{\pi}(\theta+2\pi k)\Big)
\end{equation}
for $\theta\in(-\pi,\pi)$.
\item (Uniformity after wrapping) Let $U\sim \mathrm{Unif}([-B,B])$ be independent of $Z\sim\mathcal N(0,1)$. Then $\langle U+Z \rangle_B$ is uniform on $(-\pi,\pi)$:
\begin{equation}\label{eq:wrap_unif}
f_{ \langle U+Z \rangle_B}(\theta)=\frac{1}{2\pi},\qquad \theta\in(-\pi,\pi).
\end{equation}
\item (Wrapped-mixture lower bound)  Let $X\in[-A,A]$ be discrete with $K:=|\supp(P_X)|\ge 2$, and let $U\sim\mathrm{Unif}([-A,A])$. Then
\begin{equation}\label{eq:wang_chi2}
\chiSq\!\Big(P_{\langle X+Z\rangle_A}\,\Big\|\,P_{\langle U+Z\rangle_A}\Big)\ \ge\ \frac{1}{2}\exp\!\Big(-4\pi^2\frac{K^2}{A^2}\Big).
\end{equation}
\item  (Uniform $L^\infty$ bound for the wrapped density) Let $Y=X+Z$ and assume $A\ge 1$ and $D(P_Y\|P_{Y^\star})\le \varepsilon$ with $\varepsilon\le 1/A$.
Then there is an explicit absolute constant $M_0$ such that
\begin{equation}\label{eq:wrapped_Linfty}
\|f_{\langle X+Z\rangle_A}\|_\infty\ \le\ M_0.
\end{equation}
One may take
\begin{equation}\label{eq:M0_def}
M_0
:=\Big(\frac{3}{\pi}+\frac{1}{\sqrt{2\pi}}\Big)\,\frac{2e^2}{\alpha_0}\,\pi e^2,
\quad \alpha_0=\int_{-1}^1 e^{-t^2/2}\,dt.
\end{equation}

\item Let $Y=X+Z$ and assume $A\ge 1$ and $D(P_Y\|P_{Y^\star})\le \varepsilon$ with $\varepsilon\le 1/A$. Then,
\begin{align}
    &C_0 \chiSq\!\Big(P_{\langle X+Z\rangle_A}\,\Big\|\,P_{\langle U+Z\rangle_A}\Big) \nonumber \\
    &\le     \TV(P_{\langle X+Z\rangle_A},P_{\langle U+Z\rangle_A})
\end{align}
where $C_0 = \frac{1}{4(2\pi M_0+1)}$. 
    \end{enumerate}
\end{prop}  
\begin{proof}
    The proof of the first three statements can be found in \cite{wang2025improvedlowerboundcardinality}. The last two statements are shown in Appendix~\ref{AppA}. 
\end{proof}

Conceptually, the wrapping argument allows us to reduce the problem to approximating the uniform distribution on the circle using wrapped Gaussian mixtures. 
Since uniformity is highly structured, this provides a robust way to obtain lower bounds on the number of support points.

\section{Main Result}
\label{sec:main_results}

\subsection{Some Basic Properties of $K_\varepsilon(A)$}

In this section, we collect several basic but important structural properties of $K_\varepsilon(A)$. 
In particular, we show that $K_\varepsilon(A)$ behaves monotonically in $\varepsilon$ and recovers the support size of the capacity-achieving input in the limit $\varepsilon \to 0$.

\begin{thm}[Basic properties of $K_\varepsilon(A)$]
\label{thm:basic_props_K_eps}
Fix $A>0$. Then the following statements hold.

\begin{enumerate}
    \item  For every $\varepsilon>0$, the quantity $K_\varepsilon(A)$ is well-defined and finite. In particular,
    \[
    1 \le K_\varepsilon(A) \le |\supp(P_{X^*})| < \infty,
    \]
    where $P_{X^*}$ is the capacity-achieving input distribution.

    \item  The map $\varepsilon \mapsto K_\varepsilon(A)$ is non-increasing on $(0,\infty)$.

    \item  Choose an integer $1 \le m < |\supp(P_{X^*})| $ and let 
\begin{align}
    C_{-m}(A):=
\sup\Big\{
& I(X;Y):\ \supp(P_X)\subset[-A,A],\nonumber \\
& P_X \text{ discrete}, \nonumber \\
& |\supp(P_X)|\le |\supp(P_{X^*})|-m
\Big\},
\end{align}
and let 
\begin{equation}
    \delta_m(A) = C(A) - C_{-m}(A).
\end{equation}
Then, for every $0<\varepsilon<\delta_m(A)$
\begin{equation}
    K_\varepsilon(A) \ge |\supp(P_{X^*})|-m +1, \qquad .
\end{equation}
Consequently, by choosing $m=1$
\begin{align}
\lim_{\epsilon \to 0} K_\varepsilon(A)  =  |\supp(P_{X^*})|.
\end{align}

\end{enumerate}
\end{thm}
\begin{proof}
We only show the last statement. Let $K^\star(A):=|\supp(P_{X^\star})|$. The maximizer \(P_{X^\star}\) is unique~\cite{smith1971information}. Fix an integer \(1\le m<K^\star(A)\), and set
\begin{equation}
P_X=\sum_{i=1}^{M} p_i\delta_{x_i},
\quad
(p_1,\ldots,p_M)\in\Delta_M,
\quad
|x_i|\le A,    
\end{equation}
where $M:=K^\star(A)-m$ and $\Delta_M$ is the probability simplex. Since \(\Delta_M\times[-A,A]^M\) is compact and \(I(X;X+Z)\) is continuous in the parameters \((p_1,\ldots,p_M,x_1,\ldots,x_M)\), the supremum defining \(C_{-m}(A)\) is attained. Suppose, by contradiction, that $C_{-m}(A)=C(A)$; Then a maximizer in the definition of \(C_{-m}(A)\) would also be capacity-achieving. By uniqueness, it would be \(P_X = P_{X^\star}\). However, $|\supp(P_X)|= 
M<K^\star(A)$
leads to a contradiction. Hence
$\delta_m(A)>0$.

Now, for every \(0<\varepsilon<\delta_m(A)\), we have $C(A)-\varepsilon
>
C(A)-\delta_m(A)
=
C_{-m}(A)$. 

On the other hand, by definition of $C_{-m}(A)$, every $P_X$ with
$|\supp(P_X)|\le K^*(A)-m$
satisfies
\begin{equation}
    I(X;Y)\le C_{-m}(A)<C(A)-\varepsilon.
\end{equation}
Thus, no such input can be $\varepsilon$-capacity-achieving. It follows that any $\varepsilon$-capacity-achieving input must satisfy
\begin{equation}
    |\supp(P_X)|\ge K^*(A)-m+1.
\end{equation}
By the definition of $K_\varepsilon(A)$, this implies
\begin{equation}
    K_\varepsilon(A)\ge K^*(A)-m+1=|\supp(P_{X^*})|-m+1.
\end{equation}
This proves that for every $0<\varepsilon<\delta_m(A)$,
\begin{equation}
    K_\varepsilon(A)\ge |\supp(P_{X^*})|-m+1.
\end{equation}
This concludes the proof. 
\end{proof}

\begin{rem}
Theorem~\ref{thm:basic_props_K_eps} shows that $K_\varepsilon(A)$ interpolates between two regimes. 
For large $\varepsilon$, small support sizes suffice, while as $\varepsilon \downarrow 0$, the quantity $K_\varepsilon(A)$ recovers the exact support size of the capacity-achieving input.  This motivates studying the regime where $\varepsilon \downarrow 0$ as $A$ grows.
Moreover, the quantity $\delta_m(A)$ quantifies how much capacity is lost when restricting to inputs with fewer than $|\supp(P_{X^*})|-m$ points.
\end{rem}

\subsection{Bounds}

We now state the main results of this work, which characterize the scaling of $K_\varepsilon(A)$ in different regimes of the capacity gap.

\begin{thm} \label{thm:main_result}
Suppose that $A>1600$. Then the following bounds hold.

\begin{itemize}
    \item \emph{Polynomial capacity gap.} For $\beta \ge 1$,
    \begin{align}
        \frac{1}{2\sqrt{2}\pi} A\sqrt{\log^{+}(c_1A)} 
        &\le\;  K_{ A^{-\beta}}(A) \nonumber\\
        &\le\; 32 \rme A \sqrt{c_2\log(A)}.
    \end{align}

    \item \emph{Exponential capacity gap.}
    \begin{equation}
        \frac{1}{2\sqrt{2}\pi} A\sqrt{\log^{+}(c_1A)} 
        \;\le\; K_{\rme^{-A}}(A) 
        \;\le\; c_3 A^{3/2}.
    \end{equation}
\end{itemize}

The constants are given by
\begin{equation}\label{eq:c1}
c_1 =  \frac{1}{8\left(1+\sqrt{\frac{\pi \rme}{2}}\right)\left( \Big(6+\sqrt{2\pi}\Big)\,\frac{2\pi \rme^4}{\alpha_0}+1\right)^2}, 
\end{equation}
\begin{equation}\label{eq:alpha_0}
\alpha_0=\int_{-1}^1 \rme^{-t^2/2}\,dt,    
\end{equation}
\begin{equation}
c_2=\frac{(2\beta+5) \rme}{4\log(2)+3}, 
\quad 
c_3 = 
\sqrt{\frac{12\kappa}{\log\kappa}}
+
\frac{3\sqrt{\kappa}}{1600}
+
\frac{1}{1600^{3/2}}
.
\end{equation}

\end{thm}

We make the following remarks: 
\begin{itemize}
    \item For polynomially decaying gaps, the upper and lower bounds match up to constants, yielding the characterization
    \[
    K_{A^{-\beta}}(A) = \Theta\big(A\sqrt{\log A}\big).
    \]
    In particular, the scaling is independent of $\beta \ge 1$ at the level of first-order asymptotics.

    \item The lower bound is universal across both regimes and reflects an intrinsic limitation: even moderately accurate approximation of the optimal output requires at least $A\sqrt{\log A}$ mass points.

    \item In contrast, the upper bounds reveal a  possible phase transition: while polynomial accuracy can be achieved with $A\sqrt{\log A}$ points, exponentially small gaps might require significantly larger support, up to order $A^{3/2}$.

    \item Taken together, these results show that the apparent diversity of scaling laws reported in the literature can be explained by the rate at which $\varepsilon$ decays with $A$.
\end{itemize}

\subsection{Achievability}

In this section, we demonstrate the achievability part of the main result.  We begin by presenting a general achievability bound. \begin{thm}[Achievability bound]\label{thm:eps_le_1_over_A}
		Fix $A\ge 1$ and $0<\varepsilon\le 1$. Then,
        \begin{equation}
            K_\epsilon(A) \le m
        \end{equation}
        where
\begin{equation}\label{eq:m_def_eps_le}
			m :=
			\begin{cases}
				m_1, & m_1\le \kappa A^2,\\
				m_2, & m_1> \kappa A^2.
			\end{cases}
		\end{equation} 
        and 
    \begin{align}
			m_1&=\left\lceil 3\sqrt{\kappa A}\ +\ A\sqrt{\frac{1}{c}\log\!\Big(\frac{1}{\delta_A}\Big)}\right\rceil, \label{eq:m1_eps_le} \\
			m_2&=\left\lceil \max\{3,\kappa A^2\}\ +\ A^2\log\!\Big(\frac{1}{\delta_A}\Big)\right\rceil, \label{eq:m2_eps_le} \\     
			\delta_A&= \min\left\{\frac{\varepsilon}{2},\frac{\varepsilon^2}{40(1+A^2)^2}\right\}. \label{eq:delta_eps_le}
		\end{align}
        where $c=\frac{\log\kappa}{4\kappa}$ and $\kappa$ is defined in Lemma~\ref{lem:MWY}. 
\end{thm}
\begin{proof}
    Let $P^\star$ be capacity-achieving on $[-A,A]$ and write $f^\star:=f_{P^\star}$.
		We seek to apply Lemma~\ref{lem:MWY}, target $P=P^\star$ and $m$ chosen as in \eqref{eq:m_def_eps_le}. 
		
		\emph{Case  $m_1\le \kappa A^2$:} From \eqref{eq:m_def_eps_le},  $m=m_1$. Moreover, from  \eqref{eq:m1_eps_le},  we have that $m=m_1\ge 3\sqrt{\kappa A}$, and Lemma~\ref{lem:MWY} (quadratic regime) yields
		a $Q$ supported on at most $m$ points with
		\begin{equation}
		\chiSq(f_Q\|f^\star)\le \exp\!\Big(-c\frac{m^2}{A^2}\Big).
		\end{equation}
		Since by \eqref{eq:m1_eps_le} we also have  $m=m_1\ge A\sqrt{\frac{1}{c}\log(1/\delta_A)}$, we get that 	\begin{equation}\label{eq:chi2_case1_again}
			\chiSq(f_Q\|f^\star)\le \exp\!\Big(-c\frac{m_1^2}{A^2}\Big)\le \delta_A.
		\end{equation}

    \emph{Case  $m_1>\kappa A^2$:} From \eqref{eq:m_def_eps_le}, $m=m_2$. Moreover, from \eqref{eq:m2_eps_le}, we have that  $m=m_2 \ge \max\{3,\kappa A^2\}$, in particular $m\ge 3$ and therefore $\log m\ge \log 3\ge 1$.
		Lemma~\ref{lem:MWY} (large-$m$ regime) yields a $Q$ supported on at most $m$ points such that
		\[
		\chiSq(f_Q\|f^\star)\le \exp\!\Big(-\frac{m\log m}{A^2}\Big)\le \exp\!\Big(-\frac{m}{A^2}\Big).
		\]
		Since,  by \eqref{eq:m2_eps_le} $m=m_2\ge A^2\log(1/\delta_A)$, we have that 
		\begin{equation}\label{eq:chi2_case2_again}
			\chiSq(f_Q\|f^\star)\le \exp\!\Big(-\frac{m}{A^2}\Big)\le \delta_A.
		\end{equation}
		
		In both cases, we have produced a discrete $Q$ supported on at most $m$ points such that
		\begin{equation}\label{eq:chi2_final_again}
			\chiSq(f_Q\|f^\star)\le \delta_A.
            \end{equation}

    To couple $\delta_A$ and $\varepsilon$, let $X_Q$ be the distribution achieving the bound in \eqref{eq:chi2_final_again}.  
    Also, let  $Y_Q=X_Q+Z$, with density $f_Q$ and note that 
		\begin{align}
		&C(A)-I(X_Q;Y_Q)=h(f^\star)-h(f_Q) \\
          & \le\ \|\log f^\star\|_{L^2(f^\star)} \sqrt{\chiSq(f_Q\|f^\star)}\ +\ \chiSq(f_Q\|f^\star) \label{eq:use_continuity_bound}\\
          & \le  \sqrt{10}(1+A^2) \sqrt{\chiSq(f_Q\|f^\star)}\ +\ \chiSq(f_Q\|f^\star)\label{eq:bound_on_nomr} \\
          & \le  \sqrt{10}(1+A^2) \sqrt{\delta_A}\ + \delta_A \\
          & \le \frac{\varepsilon}{2}+\frac{\varepsilon}{2}=\varepsilon , \label{eq:using_delta_choice}
		\end{align}
        where \eqref{eq:use_continuity_bound} follows from Lemma~\ref{lem:entropy_stability};  \eqref{eq:bound_on_nomr} follows from Lemma~\ref{lem:logL2_bound};  and \eqref{eq:using_delta_choice} follows from the choice of $\delta_A$ in \eqref{eq:delta_eps_le}. 
      
       Therefore,  $X_Q$ is feasible in \eqref{eq:Keps_def} with $|\supp(X_Q)|\le m$. This concludes the proof. 
\end{proof}

With Theorem~\ref{thm:eps_le_1_over_A} at our disposal, we now show the two regimes of Theorem~\ref{thm:main_result}.

\paragraph*{Achievability for $K_{ \frac{1}{A^\beta}}(A)$} In Theorem~\ref{thm:eps_le_1_over_A}, let  $\varepsilon = \frac{1}{A^\beta} $ and note that for  $A\ge 1$ 
		\begin{equation}
		\frac{1}{\delta_A}=40A^{2\beta}(1+A^2)^2\le 40A^{2\beta}(2A^2)^2=160A^{2\beta+4},
		\end{equation}
		hence $\log(1/\delta_A)\le \log 160 + (2\beta+4)\log A$. Substituting into the expression for $m_1$ in \eqref{eq:m1_eps_le} and using $1/c=4\kappa/\log\kappa$
        \begin{align}
            m_1  &\le\ 3\sqrt{\kappa A} \nonumber\\
            &\quad +A\sqrt{\frac{4\kappa}{\log\kappa}\Big(\log 160 + (2\beta+4)\log A\Big)}\ +\ 1 \\
            &  \le  C_\beta\,A\sqrt{\log A}   =: B_\beta(A), \label{eq:m1_poly_B}
		\end{align}
        where in the last inequality we have used that $\log 160 + (2\beta+4)\log A \le (2\beta+5)\log A$ for all $A\ge 160$
		and let  $C_\beta:=4\sqrt{\frac{4(2\beta+5)\kappa}{\log\kappa}}$ (the last inequality holds because $A\sqrt{\log A}$ dominates $\sqrt{A}$ and $1$ for $A\ge \rme$).   { 
We now distinguish two cases: (i) If $m_1\le \kappa A^2$, then the first branch of
Theorem~\ref{thm:eps_le_1_over_A} applies, and hence
\begin{equation}
    K_{A^{-\beta}}(A)
    \le m_1
    \le B_\beta(A);
\end{equation}
(ii) Suppose instead that $m_1>\kappa A^2$. 
Since $P_{X^\star}$ is
$\varepsilon$-capacity-achieving for every $\varepsilon>0$, we have
\begin{align}
    K_{A^{-\beta}}(A)
    &\le |\supp(P_{X^\star})| \le a_2A^2+a_1A+a_0 \label{eq:quadratic_bound}  \\
    &\le (a_2+a_1+a_0)A^2 \label{eq:A_greater_1} \\
    &<67A^2 \label{eq:use_values_a} \\
    &<16\rme^3 A^2    \le \kappa A^2 < B_\beta(A),
\end{align}
where \eqref{eq:quadratic_bound} follows from~\cite{GaussianBoundsCard}; \eqref{eq:A_greater_1} follows from $A \ge 1$; \eqref{eq:use_values_a} follows from using the numerical values of $a_0, a_1, a_2$ of~\cite{GaussianBoundsCard}; and the last inequality follows from \eqref{eq:m1_poly_B} and the assumption $m_1 > \kappa A^2$.

Thus, in either case, $K_{A^{-\beta}}(A)
    \le C_\beta A\sqrt{\log A}$.
}
        
    \paragraph*{Achievability for $K_{ \rme^{-A} }(A)$} In Theorem~\ref{thm:eps_le_1_over_A} let  $\varepsilon=\rme^{-A}$.
		To bound $\log(1/\delta_A)$, use $1+A^2\le 2A^2$ for $A\ge 1$:
		\begin{equation}
		\frac{1}{\delta_A}=40\rme^{2A}(1+A^2)^2\le 40 \rme^{2A}(2A^2)^2=160 \rme^{2A} A^{4},	    
		\end{equation}
		hence $\log(1/\delta_A)\le \log 160 + 4\log A +2A$. Substituting into the expression for $m_1$ in \eqref{eq:m1_eps_le} and using $1/c=4\kappa/\log\kappa$
        \begin{align}
          m_1  &\le\ 3\sqrt{\kappa A} \nonumber \\
          &\quad+\ A\sqrt{\frac{4\kappa}{\log\kappa}\Big(\log 160 + 4\log A + 2A\Big)}\ +\ 1\\
          & \le 
3\sqrt{\kappa A}
+
\sqrt{\frac{12\kappa}{\log\kappa}}A^{3/2}
+1   \label{eq:use_boundlog}  \\
&\le 
\left(
\sqrt{\frac{12\kappa}{\log\kappa}}
+
\frac{3\sqrt{\kappa}}{1600}
+
\frac{1}{1600^{3/2}}
\right)A^{3/2} \\
&  =:c_3 A^{3/2} 
        \end{align}       
        where \eqref{eq:use_boundlog} follows from $\log160+4\log A\le A$ for $A \ge 1600$.

For $\kappa=16\rme^3$, the constant $c_3$ is strictly
smaller than $\kappa$. Hence, for $A\ge1600$,
\begin{equation}
m_1\le c_3A^{3/2}<\kappa A^2.
\end{equation}
Thus the first branch of
Theorem~\ref{thm:eps_le_1_over_A} applies and gives
\begin{equation}
K_{\rme^{-A}}(A)
\le m_1
\le c_3A^{3/2}.
\end{equation}

\subsection{Converse}    

We now show the converse bound. 

\begin{thm}
Let $A\ge 1$ and let $0<\varepsilon\le 1/A$. Then 
\begin{equation}\label{eq:Keps_bound_general}
K_\epsilon(A) \ge\ \frac{A}{2\pi}\sqrt{\log^+\!\Bigg(\frac{c_L}{\sqrt{\varepsilon/2}+\sqrt{b_0/(2A)}}\Bigg)},
\end{equation}
where $b_0 = \sqrt{\frac{\pi \rme}{2}}$, $M_0$ is defined in \eqref{eq:M0_def}  and
\begin{equation}\label{eq:cL_def}
c_L:=\frac{1}{8(2\pi M_0+1)}.
\end{equation}
\end{thm}
\begin{proof}
Assume $X$ is such that  $C (A) - I(X;Y) \le \epsilon$, which by using \cite{topsoe1967information} implies that 
\begin{equation}
    \sfD(P_Y\|P_{Y^\star}) \le C (A) - I(X;Y) \le \epsilon. \label{eq:Stability_bound_used}
\end{equation}
We also need the following bound \cite{wang2025improvedlowerboundcardinality}
\begin{equation}
    \sfD(P_{\langle U+Z\rangle_A}\|P_{\langle Y^\star\rangle_A}) \le \frac{b_0}{A} \label{eq:stability_under_uniform}
\end{equation}
Now, 
\begin{align}
 &\frac{C_0}{2}\exp\!\Big(-4\pi^2\frac{K^2}{A^2}\Big) \le  C_0 \chiSq\!\Big(P_{\langle X+Z\rangle_A}\,\Big\|\,P_{\langle U+Z\rangle_A}\Big)  \label{eq:CHi_lower_bound} \\
 &\le     \TV(P_{\langle X+Z\rangle_A},P_{\langle U+Z\rangle_A}) \label{eq:Chi_TV_bound}\\
&\le   \TV(P_{\langle X+Z\rangle_A},P_{\langle Y^\star\rangle_A})+\TV(P_{\langle U+Z\rangle_A}, P_{\langle Y^\star\rangle_A}) \label{eq:using_triangl_inequality} \\
& \le \sqrt{ \frac{1}{2}  \sfD(P_{\langle X+Z\rangle_A}\|P_{\langle Y^\star\rangle_A})}  + \sqrt{\frac{1}{2} \sfD(P_{\langle U+Z\rangle_A}\|P_{\langle Y^\star\rangle_A})} \label{eq:using_pinskers_inequlaity}
  \\
  & \le \sqrt{ \frac{1}{2}        \sfD(P_Y\|P_{Y^\star})}  +  \sqrt{\frac{1}{2} \sfD(P_{U+Z}\|P_{Y^\star})} \label{Eq:data_processing_inequality} \\
 & \le \sqrt{ \frac{\epsilon}{2}   }  +   \sqrt{\frac{b_0}{2A}}, \label{eq:last_bound}
\end{align}
where \eqref{eq:CHi_lower_bound} and \eqref{eq:Chi_TV_bound} follow from Proposition~\ref{prop:summary_wrapping_operation}; \eqref{eq:using_triangl_inequality} follows from the triangular inequality; 
\eqref{eq:using_pinskers_inequlaity} follows from Pinsker's inequality; \eqref{Eq:data_processing_inequality} follows from data processing inequality; the first bound in \eqref{eq:last_bound} follows from \eqref{eq:Stability_bound_used}
and the  second bound follows from \eqref{eq:stability_under_uniform}. 
Rearrange and take logarithms  of \eqref{eq:last_bound} (using $\log^+(x):=\max\{\log x,0\}$):
\begin{equation}
4\pi^2\frac{K^2}{A^2}\ \ge\ \log^+\!\Bigg(\frac{c_L}{\sqrt{\varepsilon/2}+\sqrt{b_0/(2A)}}\Bigg).  \label{eq:second_to_last_bound}
\end{equation}
By rearranging the terms, we conclude the proof. 
\end{proof}

As a consequence of Theorem~\ref{thm:main_result} note that  since $\varepsilon\le 1/A$, then $\sqrt{\varepsilon/2}\le 1/\sqrt{2A}$ and hence
$\sqrt{\varepsilon/2}+\sqrt{b_0/(2A)}\le \sqrt{(1+b_0)/ (2A)}$.
Plugging this bound into \eqref{eq:second_to_last_bound} yields
\begin{align*}
K\ &\ge\ \frac{A}{2\pi}\sqrt{\log^+\!\Big(c_L\sqrt{\frac{2A}{1+b_0}}\Big)} \nonumber\\
& =\ \frac{A}{2\sqrt{2}\,\pi}\sqrt{\log^+\!\Big(\frac{2c_L^2}{1+b_0}\,A\Big)},    
\end{align*}
which gives the explicit scaling form
\begin{equation}\label{eq:Keps_AsqrtlogA}
K\ \ge\ \frac{A}{2\sqrt{2}\,\pi}\sqrt{\log^+\!\big(c\,A\big)},
\qquad
c:=\frac{2c_L^2}{1+b_0}.
\end{equation}
Consequently, for all $A$ large enough (so that $cA>1$),
\[
K_\varepsilon(A)\ \ge\ \frac{A}{2\sqrt{2}\,\pi}\sqrt{\log(cA)}.
\]

\section{Conclusions}
\label{sec:conclusion}
In this work, we studied the amplitude-constrained AWGN channel from the perspective of near-optimal input distributions. Rather than focusing on the exact capacity-achieving input, whose support size remains poorly understood, we introduced the quantity $K_\varepsilon(A)$, which captures the minimal support size required to achieve capacity up to an $\varepsilon$-gap. 

We showed that this relaxed formulation is significantly more tractable and admits sharp characterizations across different regimes. In particular, for polynomially decaying gaps, we established that $K_\varepsilon(A)=\Theta(A\sqrt{\log A})$, while for exponentially small gaps, the required support size increases to at most order $A^{3/2}$. 

Beyond the technical results, our approach provides a conceptual explanation for the variety of scaling laws observed in prior numerical studies. Namely, different empirical scalings can be interpreted as arising from different implicit choices of $\varepsilon$. 

Several open problems remain. In particular, it would be of interest to obtain tighter bounds in the exponential regime, as well as to better understand the behavior of the exact optimizer and its relation to $K_\varepsilon(A)$ as $\varepsilon \to 0$. More broadly, the $\varepsilon$-capacity perspective may prove useful in other settings where exact structural characterization is difficult but near-optimal behavior is more accessible.

 \appendices
\section{Proof of Proposition~\ref{prop:summary_wrapping_operation}}\label{AppA}

\subsection{Proof of Property 4)}

We start with some helper lemmas. 
\begin{lem}\label{lem:binary_KL_ineq}
For any $p,q\in(0,1)$,
\begin{equation}\label{eq:binary_lower}
d(p\|q):=p\log\frac{p}{q}+(1-p)\log\frac{1-p}{1-q}\ \ge\ p\log\frac{p}{q}+q-p.
\end{equation}
\end{lem}
\begin{proof}
The proof follows from inequality:
for $u\in(0,1)$,
\begin{equation}\label{eq:log_ineq}
\log(1-u)\ \ge\ -\frac{u}{1-u}.
\end{equation}
\end{proof}

\begin{lem}\label{lem:peak_stability}
Let $Y=X+Z$ and $Y^\star=X^\star+Z$, where $X$ and $X^\star$ are supported on $[-A,A]$.
Assume $\sfD(P_Y\|P_{Y^\star})\le \varepsilon$.
Let $M:=\|f_Y\|_\infty$ and define
\begin{equation}
\alpha_0:=\int_{-1}^1 e^{-t^2/2}\,dt.
\end{equation}
Then,
\begin{equation}\label{eq:peak_stability}
M\ \le\ \max\Big\{\frac{\varepsilon}{\alpha_0},\ \frac{2e^2}{\alpha_0}\,\|f_{Y^\star}\|_\infty\Big\}.
\end{equation}
Moreover,  for $A \ge 1$, $\|f_{Y^\star}\|_\infty \le \frac{\rme}{2A}$ as shown in \cite[Prop.~1]{wang2025improvedlowerboundcardinality}. 
\end{lem}
\begin{proof}
Let $y_0$ maximize $f_Y$, so $f_Y(y_0)=M$. By \cite[Eq.~(56)]{wang2025improvedlowerboundcardinality}, for $y\in[y_0-1,y_0+1]$,
\begin{equation}
f_Y(y)\ \ge\ M e^{-(y-y_0)^2/2}.
\end{equation}
Hence, with $I:=[y_0-1,y_0+1]$,
\begin{equation}
p:=P_Y(I)=\int_I f_Y(y)\,dy\ \ge\ M\int_{-1}^1 e^{-t^2/2}\,dt = \alpha_0 M.
\end{equation}
Also $q:=P_{Y^\star}(I)\le 2\|f_{Y^\star}\|_\infty$ (since $I$ has length $2$).
By data processing (coarsening to the event $I$), $\sfD(P_Y\|P_{Y^\star})\ge d(p\|q)$, hence
\begin{equation}
\varepsilon\ \ge\ d(p\|q).
\end{equation}
By Lemma~\ref{lem:binary_KL_ineq}, $d(p\|q)\ge p\log(p/q)+q-p$.
Consider two cases.

\emph{Case 1: $p/q\ge e^2$.} Then $\log(p/q)\ge 2$ and thus
\begin{equation}
\varepsilon\ \ge\ p\log(p/q)+q-p\ \ge\ 2p - p = p.
\end{equation}
Therefore $p\le \varepsilon$, and since $p\ge \alpha_0 M$, we get $M\le \varepsilon/\alpha_0$.

\emph{Case 2: $p/q< e^2$.} Then $p<e^2 q\le 2e^2\|f_{Y^\star}\|_\infty$, and since $p\ge \alpha_0 M$,
\begin{equation}
M\ \le\ \frac{2e^2}{\alpha_0}\,\|f_{Y^\star}\|_\infty.
\end{equation}
Combining the two cases yields \eqref{eq:peak_stability}. \end{proof}

We now prove our final claim which is  a uniform $L^\infty$ bound for the wrapped density. 

\begin{lem}\label{lem:wrapped_Linfty}
Let $Y=X+Z$ and assume $A\ge 1$ and $\sfD(P_Y\|P_{Y^\star})\le \varepsilon$ with $\varepsilon\le 1/A$.
Then there is an explicit absolute constant $M_0$ such that
\begin{equation}\label{eq:wrapped_Linfty}
\|f_{\langle X+Z \rangle_A}\|_\infty\ \le\ M_0.
\end{equation}
One may take
\begin{equation}\label{eq:M0_def_app}
M_0
:=\Big(\frac{3}{\pi}+\frac{1}{\sqrt{2\pi}}\Big)\,\frac{\rme^3}{\alpha_0},
\qquad \alpha_0=\int_{-1}^1 e^{-t^2/2}\,dt.
\end{equation}
\end{lem}
\begin{proof}
Let $M:=\|f_Y\|_\infty$. By Lemma~\ref{lem:peak_stability},
\begin{equation}
M\ \le\ \max\Big\{\frac{\varepsilon}{\alpha_0},\ \frac{2\rme^2}{\alpha_0}\cdot \frac{\rme}{2A}\Big\}.
\end{equation}
Under $\varepsilon\le 1/A$, we have $\varepsilon/\alpha_0 \le \frac{1}{\alpha_0 A}$, hence for all $A\ge 1$,
\begin{equation}
M\ \le\ \frac{2\rme^2}{\alpha_0}\cdot \frac{\rme}{2A}.
\end{equation}
Now apply the wrapping formula \eqref{eq:wrap_pdf} with $B=A$:
\begin{equation}
f_{\langle X+Z\rangle_A}(\theta)=\frac{A}{\pi}\sum_{k\in\mathbb Z} f_Y\!\Big(\frac{A}{\pi}(\theta+2\pi k)\Big).
\end{equation}
For $k\in\{-1,0,1\}$, the argument ranges over $[-3A,3A]$, so each term is at most $M$, hence
\begin{equation}
\frac{A}{\pi}\sum_{k\in\{-1,0,1\}} f_Y\!\Big(\frac{A}{\pi}(\theta+2\pi k)\Big)\ \le\ \frac{3A}{\pi}M.
\end{equation}
For $|k|\ge 2$, note that for $\theta\in(-\pi,\pi)$,
\begin{equation}
\Big|\frac{A}{\pi}(\theta+2\pi k)\Big|\ \ge\ A(2|k|-1)\ \ge\ 3A,
\end{equation}
so $|y|-A\ge 2A(|k|-1)$ with $y=\frac{A}{\pi}(\theta+2\pi k)$. By \cite[Eq.~(75)]{wang2025improvedlowerboundcardinality},
\begin{align}
f_Y(y)\ &\le\ f_Y(A)\exp\!\Big(-\frac{(|y|-A)^2}{2}\Big)\nonumber\\
&\le\ M\exp\!\big(-2A^2(|k|-1)^2\big).
\end{align}
Therefore,
\begin{align}
&\frac{A}{\pi}\sum_{|k|\ge 2} f_Y\!\Big(\frac{A}{\pi}(\theta+2\pi k)\Big)
\nonumber\\
&\le\ \frac{A}{\pi}M\cdot 2\sum_{m=1}^\infty e^{-2A^2 m^2}
\nonumber \\
&\le\ \frac{A}{\pi}M\cdot 2\int_0^\infty e^{-2A^2 t^2}\,dt
\ =\ \frac{M}{\sqrt{2\pi}},
\end{align}
where we used the integral comparison $\sum_{m\ge 1}g(m)\le \int_0^\infty g(t)\,dt$ for decreasing $g$
and $\int_0^\infty e^{-2A^2 t^2}dt=\frac{1}{2A}\sqrt{\frac{\pi}{2}}$.
Combining the pieces,
\begin{equation}
\|f_{\langle X+Z\rangle_A}\|_\infty \le \frac{3A}{\pi}M+\frac{M}{\sqrt{2\pi}}
\le \Big(\frac{3}{\pi}+\frac{1}{\sqrt{2\pi}}\Big)\cdot \frac{2\rme^2}{\alpha_0}\,\frac{\rme}{2},
\end{equation}
which is exactly \eqref{eq:M0_def_app}. \end{proof}

\subsection{Proof of Property~5}

We start with the following lemma. 
\begin{lem}\label{lem:chi2_to_TV}
Let $P$ be a distribution on $(-\pi,\pi)$ with density $f$ satisfying $\|f\|_\infty\le M$, and let $Q$ be uniform on $(-\pi,\pi)$.
Then
\begin{equation}\label{eq:chi2_to_TV}
\TV(P,Q)\ \ge\ \frac{1}{2(2\pi M+1)}\,\chiSq(P\|Q).
\end{equation}
\end{lem}
\begin{proof}
Since $Q$ has density $1/(2\pi)$,
\begin{align}
\chiSq(P\|Q) \nonumber
&=\int_{-\pi}^\pi \frac{(f-\frac{1}{2\pi})^2}{\frac{1}{2\pi}}d\theta\nonumber\\
&=\int_{-\pi}^\pi |2\pi f-1|\cdot \Big|f-\frac{1}{2\pi}\Big|\,d\theta\nonumber\\
&\le (2\pi M+1)\int_{-\pi}^\pi \Big|f-\frac{1}{2\pi}\Big|\,d\theta\nonumber\\
&=2(2\pi M+1)\TV(P,Q),
\end{align}
which rearranges to \eqref{eq:chi2_to_TV}. \end{proof}

The proof is competed by noting that $P_{\langle U+Z\rangle_A}$ according to Property~2 is uniform on $(-\pi, \pi)$.

\bibliography{refs.bib}
\bibliographystyle{IEEEtran}

\end{document}